\documentclass{article}
\usepackage{amsmath,amssymb} 
\usepackage{graphicx}
\usepackage{xspace}
\usepackage{gastex}
\usepackage{comment}
\usepackage{times}

\usepackage{xcolor,pgf,tikz,pgffor}
\usetikzlibrary{arrows,shapes,snakes}
\usetikzlibrary{decorations.shapes}

\newcommand\IN{{\mathbb N}}
\newcommand\IZ{{\mathbb Z}}

\newcommand{\integers}{\IZ}
\newcommand{\contractname}{\mathsf{Cnt}}
\newcommand{\contractstar}[1]{\ensuremath{\contractname^*\left(#1\right)}}
\newcommand{\contract}[1]{\ensuremath{\contractname\left(#1\right)}}
\newcommand{\contracti}[2]{\ensuremath{\contractname^{#2}\left(#1\right)}}
\newcommand{\contraction}[1]{\ensuremath{\mathsf{Contraction}\left(#1\right)}}

\newcommand{\effect}[2]{\ensuremath{\mathsf{Effect}\left(#1,#2\right)}}
\newcommand{\src}[1]{\ensuremath{\mathsf{src}\left(#1\right)}}
\newcommand{\trg}[1]{\ensuremath{\mathsf{trg}\left(#1\right)}}
\newcommand{\boundleq}{\ensuremath{|\loc|\cdot (2^{(|\edges |+ 1)}+1)}}

\newcommand{\tpof}[1]{\ensuremath{\mathsf{TPath}\left(#1\right)}}
\newcommand{\rof}[2]{\ensuremath{\mathsf{Run}\left(#1,#2\right)}}

\newcommand{\loc}{\ensuremath{\mathrm{Loc}}}
\newcommand{\rmax}{\ensuremath{\mathrm{rmax}}}
\newcommand{\cmax}{\ensuremath{\mathrm{cmax}}}
\newcommand{\edges}{\ensuremath{\mathrm{Edges}}}
\newcommand{\rate}{\ensuremath{\mathrm{Rates}}}
\newcommand{\invariants}{\ensuremath{\mathrm{Inv}}}
\newcommand{\init}{\ensuremath{\mathrm{Init}}}
\newcommand{\initval}{\ensuremath{\vec{0}}}

\newcommand{\goal}{\ensuremath{\mathrm{Goal}}}

\newcommand{\rateon}[1]{\ensuremath{\mathcal{R}\left(#1\right)}}
\newcommand{\updateon}[1]{\ensuremath{\mathcal{U}\left(#1\right)}}

\newcommand{\guardson}[1]{\ensuremath{\mathcal{G}\left(#1\right)}}
\newcommand{\duration}[1]{\ensuremath{\mathsf{duration}\left(#1\right)}}
\newcommand{\first}[1]{\ensuremath{\mathsf{first}\left(#1\right)}}
\newcommand{\last}[1]{\ensuremath{\mathsf{last}\left(#1\right)}}
\newcommand{\len}[1]{\ensuremath{\left|#1\right|}}
\newcommand{\dreset}[1]{\ensuremath{\mathsf{DetReset}\left(#1\right)}}
\newcommand{\strict}[1]{\ensuremath{\mathsf{Strict}\left(#1\right)}}
\newcommand{\cbound}[1]{\ensuremath{\mathsf{CBound}\left(#1\right)}}
\newcommand{\adapt}[2]{\ensuremath{\mathsf{Adapt}\left(#1,#2\right)}}
\newcommand{\cH}{\ensuremath{\mathcal{H}}}
\newcommand{\posreal}{\ensuremath{\mathbb{R}^+}}
\newcommand{\reals}{\ensuremath{\mathbb{R}}}
\newcommand{\false}{\ensuremath{\mathbf{false}}}
\newcommand{\true}{\ensuremath{\mathbf{true}}}
\newcommand{\tb}{\ensuremath{\mathbf{T}}}
\newcommand{\thebound}{\ensuremath{\mathbf{B}}}
\newcommand{\finalbound}{\ensuremath{2|X|+(2|X|+1)\cdot \boundleq}}

\newcommand{\tick}{\ensuremath{{\sf tick}}\xspace}
\newcommand{\atick}{\ensuremath{{\cal A}_{\tick}}\xspace}

\def\abs#1{\ensuremath{\lvert #1\rvert}}
\newcommand{\nat}{\IN}
\renewcommand{\H}{\mathcal{H}}
\renewcommand{\l}{\ell}

\def\break{\penalty-1000}
\newcommand{\val}{\ensuremath{\nu}}

\newtheorem{problem}{Problem}
\newtheorem{proposition}{Proposition}
\newtheorem{lemma}{Lemma}
\newtheorem{theorem}{Theorem}
\newenvironment{proof}{\noindent{\it Proof.}\hspace*{.3cm}}{}
\newcommand{\qed}{\hfill$\Box$}

\title{On Reachability for Hybrid Automata\\ over Bounded
  Time\thanks{Work supported by the projects: $(i)$ QUASIMODO (FP7-
    ICT-STREP-214755), Quasimodo: ``Quantitative System Properties in
    Model-Driven-Design of Embedded'', {\tt
      http://www.quasimodo.aau.dk/}, $(ii)$ GASICS (ESF-EUROCORES
    LogiCCC), Gasics: ``Games for Analysis and Synthesis of
    Interactive Computational Systems'', {\tt
      http://www.ulb.ac.be/di/gasics/}, $(iii)$ Moves: ``Fundamental
    Issues in Modelling, Verification and Evolution of Software'',
    {\tt http://moves.ulb.ac.be}, a PAI program funded by the Federal
    Belgian Government, $(iv)$ the ARC project
    AUWB-2010--10/15-UMONS-3, $(v)$ the FRFC project 2.4515.11 and
    $(vi)$ a grant from the National Bank of Belgium.}  }

\author{Thomas Brihaye\thanks{Universit\'e de Mons, Belgium} \and
  Laurent Doyen\thanks{LSV, ENS Cachan \& CNRS, France} \and 
  Gilles Geeraerts\thanks{Universit\'e Libre de Bruxelles, Belgium} \and 
  Jo\"el Ouaknine\thanks{Oxford University Computing Laboratory, UK} \and
  Jean-Fran\c{c}ois Raskin\footnotemark[3] \and James Worrell\footnotemark[4]}

\begin{document}
\maketitle

\begin{abstract}
  This paper investigates the time-bounded version of the reachability
  problem for hybrid automata. This problem asks whether a given
  hybrid automaton can reach a given target location within $\tb$ time
  units, where $\tb$ is a constant rational value. We show that, in
  contrast to the classical (unbounded) reachability problem, the
  timed-bounded version is \emph{decidable} for rectangular hybrid
  automata provided only non-negative rates are allowed.  This class
  of systems is of practical interest and subsumes, among others, the
  class of stopwatch automata.  We also show that the problem becomes
  undecidable if either diagonal constraints or both negative and
  positive rates are allowed.
\end{abstract}

\section{Introduction}

The formalism of hybrid automata~\cite{AlurCHHHNOSY95} is a
well-established model for hybrid systems whereby a digital controller
is embedded within a physical environment. The state of a hybrid
system changes both through discrete transitions of the controller,
and continuous evolutions of the environment.  The discrete state of
the system is encoded by the \emph{location} $\ell$ of the automaton,
and the continuous state is encoded by \emph{real-valued variables} $X$ 
evolving according to dynamical laws constraining the
first derivative $\dot{X}$ of the variables. Hybrid automata have
proved useful in many applications, and their analysis is supported by
several tools~\cite{HyTech,Phaver}.

A central problem in hybrid-system verification is the {\em reachability problem} 
which is to decide if there exists an
execution from a given initial location $\ell$ to a given goal
location $\ell'$.  While the reachability problem is undecidable for
simple classes of hybrid automata (such as linear hybrid
automata~\cite{AlurCHHHNOSY95}), the decidability frontier of this
problem is sharply understood~\cite{HenzingerKPV98,HenzingerR00}.  For
example, the reachability problem is decidable for the class of 
initialized rectangular automata where (i)~the flow constraints,
guards, invariants and discrete updates are defined by rectangular
constraints of the form $a \leq \dot{x} \leq b$ or $c \leq x \leq d$
(where $a,b,c,d$ are rational constants),
and~(ii)~whenever the flow constraint of a variable $x$ changes
between two locations $\ell$ and $\ell'$, then $x$ is reset along the
transition from $\ell$ to $\ell'$.  Of particular interest is the
class of timed automata which is a special class of
initialized rectangular automata~\cite{AlurD94}.

In recent years, it has been observed that new decidability results
can be obtained in the setting of time-bounded verification of
real-time systems~\cite{OuaknineRW09,OuaknineW10}. Given a time bound
$\tb \in \IN$, the time-bounded verification problems consider only
traces with duration at most $\tb$. Note that due to the density of
time, the number of discrete transitions may still be unbounded.
Several verification problems for timed automata and real-time
temporal logics turn out to be decidable in the time-bounded framework
(such as the language-inclusion problem for timed
automata~\cite{OuaknineRW09}), or to be of lower complexity (such as
the model-checking problem for {\sf MTL}~\cite{OuaknineW10}).  The
theory of time-bounded verification is therefore expected to be more
robust and better-behaved in the case of hybrid automata as well.

Following this line of research, we revisit the reachability problem
for hybrid automata with time-bounded traces.  The {\em time-bounded
  reachability problem} for hybrid automata is to decide, given a time
bound $\tb \in \IN$, if there exists an execution of duration less
than $\tb$ from a given initial location $\ell$ to a given goal
location $\ell'$.  We study the frontier between decidability and
undecidability for this problem and show how bounding time alters matters
with respect to the classical reachability problem. 
In this paper, we establish the following results. First, we
show that the time-bounded reachability problem is \emph{decidable}
for non-initialized rectangular automata when 
only positive rates are allowed\footnote{This class is interesting from
  a practical point of view as it includes, among others, the class of
  stopwatch automata~\cite{CassezL00}, for which unbounded
  reachability is undecidable.}. The proof of this fact is technical
and, contrary to most decidability results in the field, does not rely
on showing the existence of an underlying finite (bi)simulation
quotient.  We study the properties of time-bounded runs and show that
if a location is reachable within $\tb$ time units, then it is
reachable by a timed run in which the number of discrete transitions
can be bounded.  This in turn allows us to reduce the time-bounded
reachability problem to the satisfiability of a formula in the
first-order theory of real addition, decidable in
\textsf{EXPSPACE}~\cite{FR75}.

Second, we show that the time-bounded reachability problem is
\emph{undecidable} for non-initialized rectangular hybrid
  automata if both positive and negative rates are allowed.
Third, we show that the time-bounded reachability problem is
\emph{undecidable} for initialized rectangular hybrid automata
with positive singular flows if diagonal constraints in guards are allowed. 
These two undecidability results allow to precisely
characterize the boundary between decidability and undecidability.

The undecidability results are obtained by reductions from the halting
problem for two-counter machines. We present novel encodings of the
execution of two-counter machines that fit into time-bounded
executions of hybrid automata with either negative rates, or diagonal
constraints.


\section{Definitions}

Let ${\cal I}$ be the set of intervals of real numbers with endpoints in
$\integers\cup\{-\infty,+\infty\}$.
Let $X$ be a set of continuous variables, and let $X' = \{x' \mid x \in X\}$ 
and $\dot{X} = \{\dot{x} \mid x \in X\}$
be the set of primed and dotted variables, corresponding respectively 
to variable updates and first derivatives.
A \emph{rectangular constraint} over $X$ is an expression of the 
form $x\in I$ where $x$ belongs to $X$ and $I$ to ${\cal I}$.
A \emph{diagonal constraint} over $X$ is a constraint of the form $x-y\sim c$ where
$x,y$ belong to $X$, $c$ to $\integers$, and $\sim$ is in $\{ <, \le , =, \ge, > \}$.
Finite conjunctions of diagonal and rectangular
constraints over $X$ are called {\em guards},
over $\dot{X}$ they are called {\em rate constraints},
and over $X \cup X'$ they are called {\em update constraints}.
A guard or rate constraint is {\em rectangular} if 
all its constraints are rectangular. An update constraint
is {\em rectangular} if all its constraints are either rectangular 
or of the form $x=x'$.
We denote by $\guardson{X}$, $\rateon{X}$, $\updateon{X}$ respectively
the sets of guards, rate constraints, and update constraints over $X$.


\paragraph{{Linear hybrid automata.}}
A \emph{linear hybrid automaton} (LHA) is a tuple
${\mathcal{H}} = (X, \loc, \break \edges, \rate, \break \invariants, \init)$ where
$X=\{x_1,\ldots, x_{|X|}\}$ is a finite set of continuous
\emph{variables}%
; $\loc$ is a finite set of \emph{locations}; $\edges \subseteq \loc
\times \guardson{X} \times \updateon{X} \times \loc$ is a finite set
of \emph{edges}; $\rate:\loc \mapsto\rateon{X}$ assigns to each
location a constraint on the \emph{possible variable rates};
$\invariants: \loc\mapsto\guardson{X}$ assigns an \emph{invariant} to
each location; and $\init\in \loc$ is an \emph{initial location}. For
an edge $e=(\ell, g, r, \ell')$, we denote by $\src{e}$ and $\trg{e}$
the location $\ell$ and $\ell'$ respectively, $g$ is called the
\emph{guard} of $e$ and $r$ is the \emph{update} (or \emph{reset}) of
$e$. In the sequel, we denote by $\rmax$ the maximal constant
occurring in the constraints of $\{\rate(\ell)\mid \ell\in\loc\}$

A LHA $\cH$ is \emph{singular} if for all locations $\ell$ and for all
variables $x$ of $\cH$, the only constraint over $\dot{x}$ in
$\rate(\ell)$ is of the form $\dot{x} \in I$ where $I$ is a singular
interval;
it is \emph{fixed rate} if for all variables $x$ of $\cH$ there exists
$I_x\in {\cal I}$ such that for all locations $\ell$ of $\cH$, the only
constraint on $\dot{x}$ in $\rate(\ell)$ is the constraint $\dot{x}
\in I_x$.
It is \emph{multirate} if it is not fixed rate.
It is \emph{non-negative rate} if for all variables $x$, for all
locations $\ell$, the constraint $\rate(\ell)$ implies that $\dot{x}$
must be non-negative.

\paragraph{{Rectangular hybrid automata.}}
A \emph{rectangular hybrid automaton} (RHA) is a linear hybrid automaton 
in which all guards, rates, and invariants are rectangular. 
In this case, we view each reset $r$ as a function $X'\mapsto {\cal I}
\cup\{\bot\}$ that associates to each variable $x \in X$ either an
interval of possible reset values $r(x)$, or $\bot$ when the value of
the variable $x$ remains unchanged along the transition.
When it is the case that $r(x)$ is either $\bot$ or a singular interval for each $x$, we
say that $r$ is \emph{deterministic}. 
In the case of RHA, we can also view rate constraints as functions  
$\rate : \loc \times X \rightarrow {\cal I}$ that associate to each location $\ell$ and each variable $x$ an interval of possible rates $\rate(\ell)(x)$. 
A rectangular hybrid automaton $\cH$
is \emph{initialized} if for every edge $(\ell, g, r, \ell')$ of $\cH$, 
for every $x\in X$, if $\rate(\ell)(x)\neq\rate(\ell')(x)$ then
$r(x)\neq\bot$, i.e., every variable whose rate constraint is changed must be reset. 

\paragraph{{LHA semantics.}}
A \emph{valuation} of a set of variables $X$ is a function
$\val:X\mapsto\reals$.  We further denote by $\initval$ the valuation
that assigns $0$ to each variable.

Given an LHA $\cH=(X,\loc,\edges,\rate,\invariants,\init, X)$, a
\emph{state} of $\cH$ is a pair $(\ell, \val)$, where $\ell\in \loc$
and $\val$ is a valuation of $X$. The semantics of $\cH$ is defined as
follows. Given a state $s=(\ell,\val)$ of $\cH$, an \emph{edge step}
$(\ell,\val)\xrightarrow{e}(\ell',\val')$ can occur and change the
state to $(\ell',\val')$ if $e=(\ell, g, r, \ell') \in \edges$,
$\val\models g$, $\val'(x)=\val(x$) for all $x$ s.t. $r(x)=\bot$, and
$\val'(x)\in r(x)$ for all $x$ s.t. $r(x)\neq\bot$; given a time delay
$t\in\posreal$, a \emph{continuous time step}
$(\ell,\val)\xrightarrow{t}(\ell,\val')$ can occur and change the
state to $(\ell,\val')$ if there exists a vector $r=(r_1,\ldots
r_{|X|})$ such that $r\models\rate(\ell)$, $\val'=\val+(r\cdot t)$,
and $\val + (r \cdot t') \models \invariants(\ell)$ for all $0 \leq t'
\leq t$.

A \emph{path} in $\cH$ is a finite sequence $e_1, e_2, \ldots,
e_n$ of edges such that $\trg{e_i} = \src{e_{i+1}}$ for all $1\leq i\leq n-1$. 
A \emph{cycle} is a path $e_1, e_2,\ldots,
e_n$ such that $\trg{e_n}=\src{e_1}$. A cycle $e_1, e_2,\ldots, e_n$ is
\emph{simple} if $\src{e_i}\neq\src{e_j}$ for all $i\neq j$. A
\emph{timed path} of $\cH$ is a finite sequence of the form $\pi=(t_1,
e_1), (t_2,e_2),\ldots,(t_n, e_n)$, such that $e_1,\ldots, e_n$ is a path in
$\cH$ and $t_i\in\posreal$ for all $0\leq i\leq n$. We lift the
notions of cycle and simple cycle to the timed case accordingly. Given
a timed path $\pi=(t_1, e_1), (t_2,e_2),\ldots,(t_n, e_n)$, we denote
by $\pi[i:j]$ (with $1\leq i\leq j\leq n$) the timed path
$(t_i,e_i),\ldots, (t_j, e_j)$.

A \emph{run} in $\cH$ is a sequence $s_0,
(t_0,e_0), s_1, (t_1, e_1),\ldots, (t_{n-1}, e_{n-1}), s_n$ such that:
\begin{itemize}
\item $(t_0, e_0),(t_1,e_1),\ldots, (t_{n-1},e_{n-1})$ is a timed path in
$\cH$, and 
\item for all $1\leq i<n$, there exists a state $s_i'$ of $\cH$
with $s_i\xrightarrow{t_i}s_i'\xrightarrow{e_i}s_{i+1}$.
\end{itemize}
Given a run $\rho=s_0, (t_0,e_0),\dots, s_n$, let
$\first{\rho} = s_0 = (\ell_0, \val_0)$, $\last{\rho} = s_n$, 
$\duration{\rho} = \sum_{i=1}^{n-1} t_i$, and $\len{\rho}=n+1$. 
We say that~$\rho$ is 
$(i)$ \emph{strict} if $t_i > 0$ for all $1 \leq i \leq n-1$; 
$(ii)$ \emph{$k$-variable-bounded} (for $k \in \IN$) if $\val_0(x) \leq k$ for all $x \in X$, 
and $s_i \xrightarrow{t_i} (\ell_i,\val_i)$ implies that $\val_i(x) \leq k$ for all $0 \leq i \leq n$; 
$(iii)$ \emph{$\tb$-time-bounded} (for $\tb \in \IN$) if $\duration{\rho} \leq \tb$.

Note that a unique timed path $\tpof{\rho}=(t_0,
e_0),(t_1,e_1),\ldots, (t_{n-1},e_{n-1})$, is associated to each run
$\rho=s_0, (t_0,e_0), s_1,\ldots, (t_{n-1}, e_{n-1}), s_n$.  Hence, we
sometimes abuse notation and denote a run $\rho$ with
$\first{\rho}=s_0$, $\last{\rho}=s$ and $\tpof{\rho}=\pi$ by
$s_0\xrightarrow{\pi}s$. The converse however is not true: given a
timed path $\pi$ and an initial state $s_0$, it could be impossible to
build a run starting from $s_0$ and following $\pi$ because some
guards or invariants along $\pi$ might be violated. However, if such a
run exists it is necessarily unique \emph{when the automaton is
  singular and all resets are deterministic}. In that case, we denote
by $\rof{s_0}{\pi}$ the function that returns the unique run $\rho$
such that $\first{\rho}=s_0$ and $\tpof{\rho}=\pi$ if it exists, and
$\bot$ otherwise.

\paragraph{{Time-bounded reachability problem for LHA.}}
While the reachability problem asks to decide the existence of any
timed run that reaches a given goal location, we are only interested
in runs having bounded duration.

\begin{problem}[Time-bounded reachability problem]
  Given an LHA ${\mathcal{H}} = (X,\loc, \break \edges, \break
  \rate,\invariants,\init)$, a location $\goal \in \loc$ and a time
  bound $\tb \in \IN$, the \emph{time-bounded reachability problem} is
  to decide whether there exists a finite run $\rho=(\init, \initval)
  \xrightarrow{\pi} (\goal,\cdot)$ of ${\mathcal{H}}$ with
  $\duration{\rho} \le \tb$.
\end{problem}

In the following table, we summarize the known facts regarding
decidability of the reachability problem for LHA, along with the
results on time-bounded reachability that we prove in the rest of this
paper.
Note that decidability for initialized rectangular hybrid automata
(IHRA) follows directly from~\cite{HenzingerKPV98}.  
We show
decidability for (non-initialized) RHA that only have non-negative
rates in Section~\ref{sec:decid}. The undecidability of the
time-bounded reachability problem for RHA and LHA is not a consequence
of the known results from the literature and require new proofs that
are given in Section~\ref{sec:undec-bound-time}.

\begin{center}
\begin{tabular}{ || c | c | c ||} \hline \hline
HA classes & Reachability & Time-Bounded Reachability \\
\hline
LHA & U~\cite{AlurCHHHNOSY95} & \textbf{U} (see Section~\ref{sec:undec-bound-time}) \\
RHA & U~\cite{HenzingerKPV98} & \textbf{U} (see Section~\ref{sec:undec-bound-time}) \\
non-negative rates RHA & U~\cite{HenzingerKPV98} & \textbf{D} (see Section \ref{sec:decid}) \\
IRHA & D~\cite{HenzingerKPV98} & D~\cite{HenzingerKPV98}  \\
\hline \hline
\end{tabular}
\end{center}

\paragraph{Example of time bounded reachability}
Let $\cH$ be the hybrid automaton of Fig.~\ref{fig:hybsys} with the
convention that the transition starting from $\ell_i$ and ending in
$\ell_j$ is denoted $e_{ij}$. Although not explicitly stated on the
figure, we assume that all the locations are equipped with the
invariant $(x \le 1) \wedge (y\le 1)$.  As this automaton uses only
rectangular constraints and positive rates, it is in the class for
which we show the decidability of the time-bounded reachability
problem (see Section~\ref{sec:decid}). Note that it is non-initialized
as, for example, variable $y$ is not reset from location $\ell_0$ to
location $\ell_1$ while its rate is changing, and it is singular,
diagonal-free, and multirate.

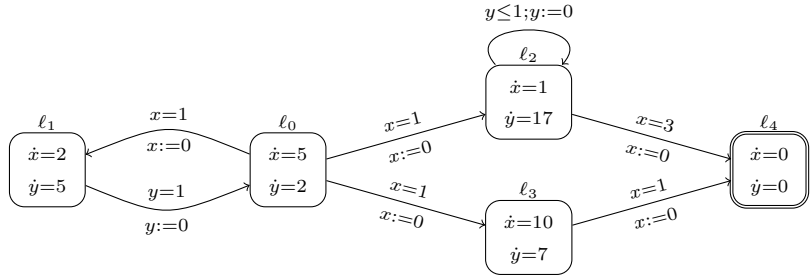
\begin{figure}[h!]
\begin{center}
\begin{tikzpicture}[xscale=.8,yscale=.9]
  \everymath{\scriptstyle}

  \path (0,0) node[draw,rectangle,rounded corners=2mm,inner sep=2pt]
  (q0) {$\begin{array}{c} \dot{x} = 5 \\ \dot{y} = 2 \end{array}$};

  \path (0,.7) node[] (q0b) {$\ell_0$};
  \path (-4,.7) node[] (q0b) {$\ell_1$};
  \path (4,1.7) node[] (q0b) {$\ell_2$};
  \path (4,-.3) node[] (q0b) {$\ell_3$};
  \path (8,.7) node[] (q0b) {$\ell_4$};


  \path (-4,0) node[draw,rectangle,rounded corners=2mm,inner sep=2pt]
  (q1) {$\begin{array}{c} \dot{x} = 2 \\ \dot{y} = 5 \end{array}$};

  \path (4,1) node[draw,rectangle,rounded corners=2mm,inner sep=2pt]
  (q2) {$\begin{array}{c} \dot{x} = 1 \\ \dot{y} = 17 \end{array}$};

  \path (4,-1) node[draw,rectangle,rounded corners=2mm,inner sep=2pt]
  (q3) {$\begin{array}{c} \dot{x} = 10 \\ \dot{y} = 7 \end{array}$};

  \path (8,0) node[draw,rectangle,rounded corners=2mm,inner sep=2pt,double]
  (q4) {${\begin{array}{c} \dot{x} = 0 \\ \dot{y} =
        0 \end{array}}$};

   \draw[arrows=->] (q0) .. controls +(160:2.1cm) .. (q1)
      node[pos=.5,above,sloped] {{$x=1$}}
      node[pos=.5,below,sloped] {${x:=0}$};

   \draw[arrows=->] (q1) .. controls +(340:2.1cm) .. (q0)
      node[pos=.5,above,sloped] {{$y=1$}}
      node[pos=.5,below,sloped] {${y:=0}$};

   \draw[arrows=->] (q0) -- (q2)
      node[pos=.5,above,sloped] {{$x=1$}}
      node[pos=.5,below,sloped] {${x:=0}$};

   \draw[arrows=->] (q0) -- (q3)
      node[pos=.5,above,sloped] {{$x=1$}}
      node[pos=.5,below,sloped] {${x:=0}$};

   \draw[arrows=->] (q2) .. controls +(135:1.75cm) and +(45:1.75cm) .. (q2)
      node[pos=.5,above,sloped] {{$y\le1;y:=0$}};

   \draw[arrows=->] (q2) -- (q4)
      node[pos=.5,above,sloped] {{$x=3$}}
      node[pos=.5,below,sloped] {${x:=0}$};

   \draw[arrows=->] (q3) -- (q4)
      node[pos=.5,above,sloped] {{$x=1$}}
      node[pos=.5,below,sloped] {${x:=0}$};  
\end{tikzpicture}
\end{center}
\caption{A singular, diagonal-free, multirate hybrid automaton.}
\label{fig:hybsys}
\end{figure}

Assume we want to reach location $\ell_4$ from $(\ell_0,0,0)$ within
one time unit. One clearly see that the duration of any run starting
from $\ell_0$ and crossing $\ell_2$ will exceed one time unit. An
other possibility would be to directly go from $\ell_0$ to $\ell_3$.
In this case, when reaching location $\ell_3$, after crossing
$e_{03}$, the value of the variable $x$ (resp. $y$) is $0$
(resp. $\frac{2}{5}$). Thus, in order to cross $e_{34}$, one should
wait $\frac{1}{10}$ time units, if we do so, the value of $y$ will
reach $\frac{11}{10}$ and violate the invariant. It is thus impossible
to reach $\ell_3$ from $(\ell_0,0,0)$ without visiting $\ell_1$. A
single visit to $\ell_1$ is sufficient as the following run testifies:
$ (\ell_0,0,0) \xrightarrow{\frac{1}{5},e_{01}}
\left(\ell_1,0,\frac{2}{5}\right) \xrightarrow{\frac{3}{25},e_{10}}
\left(\ell_0,\frac{6}{25},0\right) \xrightarrow{\frac{17}{125},e_{03}}
\left(\ell_3,0,\frac{34}{125}\right) \xrightarrow{\frac{1}{10},e_{34}}
\left(\ell_4,0,\frac{243}{250}\right).  $ The illustration of the
evolution of the variables along this run is given in
Fig.~\ref{fig:evol-var-1}. In this picture, the evolution of the
$x$-variable (resp. of the $y$-variable) is represented by the dashed
(resp. plain) curve. The evolutions of the valuations of the variables
along the beginning of the unique run looping between $\ell_0$ and
$\ell_1$ is illustrated in Fig.~\ref{fig:evol-var-2}. Looking at that
looping run, one could be convinced that $\cH$ does not admit a finite
bisimulation quotient. 

\begin{figure}
  \null\hfill
  \begin{minipage}[b]{0.48\linewidth}
    \begin{center}
\begin{tikzpicture}[xscale=7]
  \everymath{\scriptstyle}

  \draw[->] (0,0)--(.65,0); 

  \draw[->] (0,0)--(0,1.2);

  \draw (0,1)--(.65,1);

  \draw[dotted] (.2,0)--(.2,1);
  \draw[dotted] (.32,0)--(.32,1);
  \draw[dotted] (.456,0)--(.456,1);
  \draw[dotted] (.556,0)--(.556,1);

  \draw[dashed] (0,0)--(.2,1);
  \draw[] (0,0)--(.2,.4)--(.32,1);

  \draw (.2,-.05)--(.2,0); 
  \path (.2,-.25) node[] (q0) {$\frac{1}{5}$};

  \draw[dashed] (.2,0)--(.32,.24) -- (.456,1);

  \draw (.32,-.05)--(.32,0); 
  \path (.32,-.25) node[] (q0) {$\frac{8}{25}$};

  \draw[] (.32,0)--(.456,.272)--(.556,.972);

  \draw (.456,-.05)--(.456,0); 
  \path (.456,-.25) node[] (q0) {$\frac{57}{125}$};

  \draw[dashed] (.456,0)--(.556,1);

  \draw (.556,-.05)--(.556,0); 
  \path (.556,-.25) node[] (q0) {$\frac{139}{250}$};
\end{tikzpicture}
      \caption{A successful run.}
      \label{fig:evol-var-1}
    \end{center}
  \end{minipage}
  \hfill
  \begin{minipage}[b]{0.48\linewidth}
    \begin{center}
\begin{tikzpicture}[xscale=7]
  \everymath{\scriptstyle}

  \draw[->] (0,0)--(.85,0); 

  \draw[->] (0,0)--(0,1.2);

  \draw (0,1)--(.85,1);

  \draw[dashed] (0,0)--(.2,1);
  \draw[dotted] (0,0)--(.2,.4)--(.32,1);

  \draw (.2,-.05)--(.2,0); 
  \path (.2,-.25) node[] (q0) {$\frac{1}{5}$};

  \draw[dashed] (.2,0)--(.32,.24) -- (.456,1);

  \draw (.32,-.05)--(.32,0); 
  \path (.32,-.25) node[] (q0) {$\frac{8}{25}$};

  \draw[dotted] (.32,0)--(.456,.272) -- (.601,1);

  \draw (.456,-.05)--(.456,0); 
  \path (.456,-.25) node[] (q0) {$\frac{57}{125}$};

  \draw[dashed] (.456,0)--(.601,.291) -- (.743,1);

  \draw (.601,-.05)--(.601,0); 
  \path (.601,-.25) node[] (q0) {$\frac{376}{625}$};

  \draw[dotted] (.601,0)--(.743,.283);

  \draw (.743,-.05)--(.743,0); 
  \path (.743,-.25) node[] (q0) {$\frac{2323}{3125}$};
  
\end{tikzpicture}
\caption{A loop between $\ell_0$ and $\ell_1$.}
      \label{fig:evol-var-2} 
    \end{center}
  \end{minipage}
  \hfill\null
\end{figure}


\section{Decidability for  RHA with Non-Negative Rates}\label{sec:decid}
In this section, we prove that the time-bounded reachability problem
is \emph{decidable} for the class of (non-initialized)
\emph{rectangular} hybrid automata having \emph{non-negative rates},
while it is \emph{undecidable} for this class in the classical
(unbounded) case~\cite{HenzingerKPV98}. Note that this class is
interesting in practice since it contains, among others, the
important class of \emph{stopwatch automata}, a significant subset of
LHA that has several useful applications~\cite{CassezL00}.  We obtain
decidability by showing that for RHA with non-negative
rates, a goal location is reachable within $\tb$ time units iff there
exists a witness run of that automaton which reaches the goal (within
$\tb$ time units) by a run~$\rho$ of length $\abs{\rho} \leq
K_\tb^\cH$ where $K_\tb^\cH$ is a parameter that depends on $\tb$ and
on the size of the automaton~$\cH$. Time-bounded reachability can thus
be reduced to the satisfiability of a formula in the first order
theory of the reals encoding the existence of runs of length at most
$K_\tb^\cH$ and reaching $\goal$.

For simplicity of the proofs, we consider RHA with the following
restrictions: (i)~the guards \emph{do not contain strict
  inequalities}, and (ii)~the rates are \emph{singular}.  We argue at
the end of this section that these restrictions can be made without
loss of generality.  Then, in order to further simplify the
presentation, we show how to syntactically simplify the automaton
while preserving the time-bounded reachability properties. The details
of the constructions can be found in the appendix.

\begin{proposition}\label{prop:syntax}
  Let $\cH$ be a singular RHA with non-negative rates and without
  strict inequalities, and let $\goal$ be a location of $\cH$. We can
  build a hybrid automaton $\cH'$ with the following the properties:
  \begin{enumerate}
  \item[$\sf H_1$] $\cH'$ is a singular RHA with non-negative rates
  \item[$\sf H_2$] $\cH'$ contains only deterministic resets
  \item[$\sf H_3$] for every edge $(\ell,g,r,\ell')$ of $\cH'$, $g$ is
    either $\true$ or of the form $x_1=1\wedge x_2=1\wedge\cdots\wedge
    x_k=1$, and $r\equiv x_1'=0\wedge \cdots \land x_k'=0$.
  \end{enumerate}
  and a set of locations $S$ of $\cH'$ such that $\cH$ admits a
  $\tb$-time bounded run reaching $\goal$ iff $\cH'$ admits a
  strict $1$-variable-bounded, and $\tb$-time bounded run reaching
  $S$.
\end{proposition}
\begin{proof}
  The proof is given in Appendix~\ref{sec:constr-prove-prop}\qed
\end{proof}
As a consequence, to prove decidability of time-bounded reachability
of RHA with non-negative rates, we only need to prove that we can
decide whether an RHA respecting $\sf H_1$ through $\sf H_3$ admits a
\emph{strict} run $\rho$ reaching the goal within $\tb$ time units,
and where all variables are bounded by $1$ along $\rho$.

\paragraph{Bounding the number of equalities.} As a first step to
obtain a witness of time-bounded reachability, we bound the number of
transitions guarded by equalities along a run of bounded duration:

\begin{proposition}\label{prop-equal}
  Let $\cH$ be an LHA, with set of variables $X$ and respecting
  hypothesis $\sf H_1$ through $\sf H_3$. Let $\rho$ be a $\tb$-time
  bounded run of $\cH$. Then, $\rho$ contains at most
  $|X|\cdot\rmax\cdot\tb$ transitions guarded by an equality.
\end{proposition}
\begin{proof}
  For a contradiction, assume that there exists an execution $\rho$ of
  ${\mathcal{H}}$ with $M$ transitions containing (at least) an
  equality where $M > |X| \cdot \rmax \cdot \tb$. By $\sf H_3$, the
  equalities in the guards are of the form $x=1$. In particular, there
  must exists a variable $y \in X$ which has been tested equal to one
  (and thus reset to zero by $\sf H_3$) strictly more than $\rmax \cdot
  \tb$ times. Since all the rates of $y$ are non negative by $\sf
  H_1$, the shortest time needed to reach the guard $y=1$ from the
  value $0$ is $\frac{1}{\rmax}$. Along $\rho$, the variable $y$ has
  reached the guard $y=1$ from $0$ strictly more than $\rmax \cdot
  \tb$ times; this implies that $\duration{\rho} > \rmax \cdot \tb
  \cdot \frac{1}{\rmax} = \tb$ which is a contradiction.\qed
\end{proof}



\paragraph{Bounding runs without equalities.}
Unfortunately, it is not possible to bound the number of transitions
that do not contain equalities, even along a time-bounded
run. However, we will show that, given a time-bounded run $\rho$
without equality guards, we can build a run $\rho'$ that is equivalent
to $\rho$ (in a sense that its initial and target states are the
same), and whose length is \emph{bounded} by a parameter depending on
the size of the automaton. More precisely:
\begin{proposition}\label{prop:bounded-length}
  Let $\cH$ be an RHA with non-negative rates. For any $1$-variable
  bounded and $\frac{1}{\rmax+1}$-time bounded run
  $\rho=s_0\xrightarrow{\pi}s$ of $\cH$ that contains no equalities in
  the guards, $\cH$ admits a $1$-variable bounded and
  $\frac{1}{\rmax+1}$-time bounded run $\rho'=s_0\xrightarrow{\pi'}s$
  such that $\len{\rho'}\leq \finalbound$.
\end{proposition}

Note that Proposition~\ref{prop:bounded-length} applies only to runs of duration at most
$\frac{1}{\rmax+1}$. However, this is not restrictive, since any
$\tb$-time-bounded run can always be split into at most
$\tb\cdot(\rmax+1)$ subruns of duration at most $\frac{1}{\rmax+1}$,
provided that we add a self-loop with guard $\true$ and no reset on
every location (this can be done without loss of generality as far as
reachability is concerned).

To prove Proposition~\ref{prop:bounded-length}, we rely on a
\emph{contraction operation} that receives a \emph{timed path} and
returns another one of smaller length.  Let $\pi=(t_1,e_1), (t_2,
e_2),\ldots, (t_n,e_n)$ be a timed path. We define $\contract{\pi}$ by
considering two cases. Let $j$, $k$, $j'$, $k'$ be four positions such
that $1\leq j\leq k< j'\leq k'\leq n$ and $e_j\ldots e_k=e_j'\ldots
e_k'$ is a \emph{simple cycle}. \textbf{If} such $j$, $k$, $j'$, $k'$
exist, then let:
\begin{eqnarray*}
  \contract{\pi}&=&
  \pi[1: j-1]\cdot(e_j,t_j+t_{j'})\cdots (e_k, t_k+t_{k'})\cdot\pi[k+1:j'-1]\cdot\pi[k'+1:n]
\end{eqnarray*}
\textbf{Otherwise}, we let $\contract{\pi}=\pi$.  
Observe that $\pi$ and $\contract{\pi}$ share the same source and
target locations, even when $\pi[k'+1:n]$ is empty.

Then, given a timed path $\pi$, we let $\contracti{\pi}{0}=\pi$,
$\contracti{\pi}{i}=\contract{\contracti{\pi}{i-1}}$ for any $i\geq
1$, and $\contractstar{\pi}=\contracti{\pi}{n}$ where $n$ is the least
value such that $\contracti{\pi}{n}=\contracti{\pi}{n+1}$. Clearly,
since $\pi$ is finite, and since $\len{\contract{\pi}}<\len{\pi}$ or
$\contract{\pi}=\pi$ for any~$\pi$, $\contractstar{\pi}$ always
exists. Moreover, we can always bound the length of
$\contractstar{\pi}$. This stems from the fact that
$\contractstar{\pi}$ is a timed path that contains at most one
occurrence of each simple cycle. The length of such paths can be
bounded using classical combinatorial arguments.

\begin{lemma}\label{lem:bound-contract-star}
  For any timed path $\pi$ of an LHA $\cH$ with $|\loc|$
  locations and $|\edges|$ edges: $\len{\contractstar{\pi}}\leq
  \boundleq$.
\end{lemma}
\begin{proof}
   Let $\contractstar{\pi}=(t_1, e_1), (t_2, e_2),\ldots,
  (t_n,e_,)$. First, observe that, by definition of $\contractname^*$,
  the actual values of the time delays $t_1$, $t_2$,\ldots $t_n$ are
  irrelevant to the length of $\contractstar{\pi}$, since the
  `contraction' is based solely on the edges. Still by definition of
  $\contractname^*$, also observe that the path $e_1, e_2,\ldots,
  e_n$ does not contain two occurrences of the same simple
  cycle. Thus, the length of $\contractstar{\pi}$ is always bounded by
  the length of the maximal path in $\cH$ that does not contain two
  occurrences of the same simple cycle.

  In order to compute this value, we first observe that any path
  $\sigma=e_1, e_2,\ldots e_n$ can always be decomposed into subpaths
  $\sigma_1,\sigma_2,\ldots \sigma_{2k},\sigma_{2k+1}$ where any
  $\sigma_{2i+1}$ (for $0\leq i\leq k$) is an acyclic path and any
  $\sigma_{2j}$ is a simple cycle (for $1\leq j\leq k$). This stems
  from the fact that any cycle (whether it is simple or not) can
  always be decomposed into a sequence of simple cycles and acyclic
  paths.

  Thus, the worst case scenario for a path containing at most one each
  simple cycle is to have a path of the form:
  $\sigma_1,\sigma_2,\ldots \sigma_{2k},\sigma_{2k+1}$ where each
  $\sigma_{2i+1}$ (for $0\leq i\leq k$) is of maximal length, and
  $\{\sigma_{2j} \mid 1\leq j\leq k\}$ is the set of all possible
  simple cycles. By definition of a simple cycle, in an automaton with
  $|\edges|$ and $|\loc|$ locations, there are at most $2^{|\edges|}$
  simple cycles, and each of them has at most length $|\loc|$
  (otherwise the cycle would contain two edges with the some origin
  and the cycle wouldn't be simple). Moreover, in such an automaton,
  each acyclic path is of length at most $|\loc|$ too. Hence, the
  worst case is a path $\sigma_1,\sigma_2,\ldots
  \sigma_{2k},\sigma_{2k+1}$ where, $k=2^{|\edges|}$, for all $1\leq
  i\leq k$: $\len{\sigma_{2i}}=|\loc|$ and for all $0\leq j\leq k$:
  $\len{\sigma_{2j+1}}=|\loc|$, that is a total length of $k\cdot
  |\loc| + (k+1)\cdot |\loc| = |\loc|\cdot (2k+1)=\boundleq$.\qed
\end{proof}

Note that the contraction operation is purely syntactic and works on
the timed path only. Hence, given a run $s_0\xrightarrow{\pi}s$, we
have no guarantee that
$\rof{s_0}{\contractstar{\pi}}\neq\bot$. Moreover, even in the
alternative, the resulting run might be
$s_0\xrightarrow{\contractstar{\pi}}s'$ with $s\neq s'$.
Nevertheless, we can show that $\contractstar{\pi}$ preserves some
properties of $\pi$. For a timed path $\pi=(t_1,e_1),\ldots,(t_n,e_n)$
of an LHA $\cH$ with rate function $\rate$, we let
$\effect{\pi}{x}=\sum_{i=1}^n\rate(\ell_i)(x)\cdot t_i$, where
$\ell_i$ is the initial location of $e_i$ for any $1\leq i\leq n$. Note
thus that, for any run $(\ell, \val)\xrightarrow{\pi}(\ell',\val')$,
for any variable $x$ \emph{which is not reset along $\pi$},
$\val'(x)=\val(x)+\effect{\pi}{x}$. It is easy to see that
$\contractstar{\pi}$ preserves the effect of $\pi$. Moreover, the
duration of $\contractstar{\pi}$ and $\pi$ are equal.

\begin{lemma}\label{lem:duration-preserved}\label{lem:contract-preserves-effect}
  For any timed path $\pi$: $(i)$
  $\duration{\pi}=\duration{\contractstar{\pi}}$ and $(ii)$  for any variable $x$:
  $\effect{\pi}{x}=\effect{\contractstar{\pi}}{x}$.
\end{lemma}


We are now ready to show, given a timed path $\pi$ (with
$\duration{\pi}\leq \frac{1}{\rmax+1}$ and without equality tests in
the guards), how to build a timed path $\contraction{\pi}$ that fully
preserves the values of the variable, as stated in
Proposition~\ref{prop:bounded-length}. The key ingredient to obtain
$\contraction{\pi}$ is to apply $\contractname^*$ to selected portions
of $\pi$, in such a way that for each edge $e$ that resets a variable for
the \emph{first} or the \emph{last} time along $\pi$, the time
distance between the occurrence of $e$ and the beginning of the timed
path is the same in both $\pi$ and $\contraction{\pi}$. 

The precise construction goes as follows. Let $\pi=(t_1, e_1),\ldots,
(t_n, e_n)$ be a timed path. For each variable $x$, we denote by
$S_x^\pi$ the set of positions $i$ such that $e_i$ is either the
\emph{first} or the \emph{last} edge in $\pi$ to reset $x$
(hence $|S_x^\pi|\in\{0,1,2\}$ for any $x$). Then, we decompose $\pi$
as: $\pi_1\cdot (t_{i_1}, e_{i_1})\cdot \pi_2 \cdot (t_{i_2},
e_{i_2})\cdots (t_{i_k},e_{i_k}) \cdot \pi_{k+1}$ with $\{i_1,\ldots,
i_k\}=\cup_x S_x^\pi$. From this decomposition of $\pi$, we let
$\contraction{\pi}=\contractstar{\pi_1}\cdot (t_{i_1}, e_{i_1})\cdot
\contractstar{\pi_2} \cdot (t_{i_2}, e_{i_2})\cdots (t_{i_k},e_{i_k})
\cdot \contractstar{\pi_{k+1}}$.

We first note that, thanks to Lemma~\ref{lem:bound-contract-star},
$\len{\contraction{\pi}}$ is bounded.

\begin{lemma}\label{lem:length-pi-prime}
  Let $\cH$ be an LHA with set of variable $X$, set of edges $\edges$ and
  set of location $\loc$, and let $\pi$ be a timed path of $\cH$. Then
  $\len{\contraction{\pi}}\leq 2\cdot |X|+(2\cdot |X|+1)\cdot
  \boundleq$.
\end{lemma}
\begin{proof}
   The Lemma stems from the fact that $|\cup_xS^\pi_x|\leq 2\cdot |X|$
  and that, for any~$j$: $\len{\contractstar{\pi_j}}\leq \boundleq$ by
  Lemma~\ref{lem:bound-contract-star}.\qed
\end{proof}


In order to obtain Proposition~\ref{prop:bounded-length}, it remains
to show that this construction can be used to build a run $\rho'$ that
is equivalent to the original run $\rho$. By
Lemma~\ref{lem:duration-preserved}, we know that
$\duration{\contractstar{\pi_j}}=\duration{\pi_j}$ for any~$j$. Hence,
the first and last resets of each variable happen at the same time
(relatively to the beginning of the timed path) in both $\pi$ and
$\contraction{\pi}$. Intuitively, preserving the time of occurrence of
the first reset (of some variable $x$) guarantees that $x$ will never
exceed $1$ along $\contraction{\pi}$, because
$\duration{\contraction{\pi}}=\duration{\pi}\leq
\frac{1}{\rmax+1}$. Symmetrically, preserving the last reset of some
variable $x$ guarantees that the final value of $x$ will be the same in
both $\pi$ and $\contraction{\pi}$. Moreover, we know (see
Lemma~\ref{lem:contract-preserves-effect}) that the contraction
function also preserves the value of the variables that are not reset.
Thanks to these results, we are now ready to prove
Proposition~\ref{prop:bounded-length}.

\begin{proof}[of Proposition~\ref{prop:bounded-length}]
  Let $\pi=\tpof{\rho}$ and let $\pi'$ denote $\contraction{\pi}$.  To
  prove the existence of $\rho'$, we will choose
  $\rho'=s_0\xrightarrow{\pi'}s$.  Let us first show that
  $\rof{s_0}{\pi'}\neq\bot$. Since $\pi$ and $\pi'$ contain no
  equality test, by ${\sf H_3}$, this amounts to showing that firing
  $\pi'$ from $s_0$ will always keep all the variable values $\leq 1$.

  Let us consider the decomposition of $\pi$ into: $\pi_1\cdot
  (t_{i_1}, e_{i_1})\cdot \pi_2 \cdot (t_{i_2}, e_{i_2})\cdots
  (t_{i_k},e_{i_k}) \cdot \pi_{k+1}$, as in the definition of ${\sf
    Contraction}$. For any $1\leq i\leq k$, let $s_i=(\ell_i,\val_i)$
  denote the state reached by the run
  $s_0\xrightarrow{\pi_1\cdot(t_{i_1}, e_{i_1})\cdots
    \pi_i}s_i$. Symmetrically, let $s_i'=(\ell_i,\val_i')$ denote the
  state reached by the run
  $s_0\xrightarrow{\contractstar{\pi_1}\cdot(t_{i_1}, e_{i_1})\cdots
    \contractstar{\pi_i}}s_i'$, assuming it exists. In that case, we
  observe that, for any variable $x$ which is not reset along
  $\contractstar{\pi_1}\cdot(t_{i_1}, e_{i_1})\cdots
  \contractstar{\pi_i}$, we have: $\val_i(x)=\val_i'(x)$, by
  Lemma~\ref{lem:contract-preserves-effect}.

  Then, we proceed by contradiction. Let $(t_j, e_j)$ be an element
  from $\pi'$, let $x$ be a variable
  such that $s_0\xrightarrow{\pi'[1:j]}(\ell',\val')$ and
  $\val'(x)+\rate(\ell')(x)\cdot t_{j+1} >1$. We first observe that, once
  $x$ has been reset along $\pi'$, its value can never exceed $1$
  because
  $\duration{\pi'}=\duration{\pi}\leq\frac{1}{\rmax+1}$. Hence,
  $(t_j,e_j)$ must occur \emph{before} the first reset of $x$ along
  $\pi'$. We distinguish two cases:
  \begin{enumerate}
  \item In the case where $(t_j,e_j)$ occurs in some part
    $\contractstar{\pi_{i_j}}$ of the decomposition of $\pi'$, we know
    that
    $\val_{i_j-1}'(x)+\effect{(t_{i_j},e_{i_j})\contractstar{\pi_{i_j}}}{x}>1$,
    since $x$ is not reset along $\contractstar{\pi_{i_j}}$. However,
    we have:
    \begin{align*}
      \val_{i_j}(x)&= \val_{i_j-1}(x)+\effect{(t_{i_j},e_{i_j})\cdot\pi_{i_j}}{x}&\textrm{def. and $x$ not reset}\\
      &=\val_{i_j-1}'(x)+\effect{(t_{i_j},e_{i_j})\cdot\pi_{i_j}}{x}&\textrm{observation above}\\
      &=\val_{i_j-1}'(x)+\effect{(t_{i_j},e_{i_j})\cdot\contractstar{\pi_{i_j}}}{x}&\textrm{Lemma~\ref{lem:contract-preserves-effect}}\\
      &>1
    \end{align*}
    Hence, $\rho$ reaches a valuation where the value of $x$ exceeds
    $1$. Contradiction.
  \item The case where $(t_j, e_j)=(t_{i_k}, e_{i_k})$ for some $i_k$
    is treated similarly and leads to the same contradiction.
  \end{enumerate}

  Now, we are sure that $\rho'=s_0\xrightarrow{\pi'}(\ell',\val')$ is
  indeed a $1$-variable bounded run. By
  Lemma~\ref{lem:length-pi-prime}, it has the adequate length. It
  remains to show that $\rho=s_0\xrightarrow{\pi}(\ell,\val)$ implies
  $\ell'=\ell$ and $\val=\val'$. The first point is true by definition
  of $\pi'$. For any variable $x$, let $i_x$ denote the element
  $(t_{i_x},e_{i_x})$ of $\pi$ where the \emph{last} reset of $x$
  occurs along $\pi$ (and thus along $\pi'$). We observe that
  $\val(x)=\effect{\pi_{i_x+1}\cdot(t_{i_x+1},e_{i_x+1})\cdots
    \pi_{k+1}}{x}$ and that
  $\val'(x)=\effect{\contractstar{\pi_{i_x+1}}\cdot(t_{i_x+1},e_{i_x+1})\cdots\contractstar{\pi_{k+1}}}{x}$
  since $x$ is not reset anymore along those two suffixes.  By
  Lemma~\ref{lem:contract-preserves-effect}, we have
  $\val(x)=\val'(x)$.  \qed
\end{proof}

\paragraph{Handling `$<$' and non-singular rates.} Let us now briefly
explain how we can adapt the construction of this section to cope with
strict guards and non-singular rates. First, when the RHA $\cH$
contains strict guards, the RHA $\cH'$ of
Proposition~\ref{prop:syntax} will also contain guards with atoms of
the form $x<1$. Thus, when building a `contracted path' $\rho'$
starting from a path $\rho$ (as in the proof of
Proposition~\ref{prop:bounded-length}), we need to ensure that these
strict guards will also be satisfied along $\rho'$. It is easy to use
similar arguments to establish this: if some guard $x<1$ is not
satisfied in $\rho'$, this is necessarily before the first reset of
$x$, which means that the guard was not satisfied in $\rho$ either. On
the other hand, to take non-singular rates into account, we need to
adapt the definition of timed path. A timed path is now of the form
$(t_0, r_0, e_0)\cdots (t_n,r_n,e_n)$, where each $r_i$ is a vector of
reals of size $|X|$, indicating the actual rate that was chosen for
each variable when the $i$-th continuous step has been taken. It is
then straightforward to adapt the definitions of $\contractname$,
${\sf Effect}$ and ${\sf Contraction}$ to take those rates into
account and still keep the properties stated in
Lemma~\ref{lem:bound-contract-star} and~\ref{lem:length-pi-prime} and
in Proposition~\ref{prop:bounded-length} (note that we need to rely on
the convexity of the invariants in RHA to ensure that proper rates can
be found when building $\contract{\pi}$). 

\begin{theorem}\label{theo:tb-reach-decidable}
  The time-bounded reachability problem is decidable for the class of
  rectangular hybrid automata with non-negative rates.
\end{theorem}
\begin{proof}
   Let $\cH$ be an RHA with non-negative rates, let $\goal$ be one of
  its location, let $\thebound$ be a natural value, and let us show
  how to determine whether $\cH$ admits a $\thebound$-time-bounded run
  reaching $\goal$.  By Proposition~\ref{prop:syntax} (and taking into
  account the above remarks to cope with strict guards and rectangular
  rates), this amounts to determining the exists of a strict $1$-variable
  bounded run reaching $\goal'$ in $\cH'$ (where $\goal'$ and $\cH'$
  are defined as in Proposition~\ref{prop:syntax}). By
  Proposition~\ref{prop:bounded-length}, this can be done by
  considering only the runs of length at most $\finalbound$ in
  $\cH'$. This question can be answered by building an ${\sf
  FO}(\reals,\leq,+)$ formula $\varphi_{\cH'}$ which is
  satisfiable iff $\rho'$ exists. Since the satisfiability of ${\sf
  FO}(\reals,\leq,+)$ is decidable \cite{FR75}, we obtain the
  theorem.\qed
\end{proof}



\section{Undecidability Results\label{sec:undec-bound-time}}

In this section, we show that the time-bounded reachability problem
for linear hybrid automata becomes undecidable if either both positive
and negative rates are allowed, or diagonal constraints are allowed in
the guards.  Along with the decidability result of
Section~\ref{sec:decid}, these facts imply that the class of
rectangular hybrid automata having positive rates only and no diagonal
constraints forms a maximal decidable class.  Our proofs rely on
reductions from the halting problem for Minsky two-counters machines.

A \emph{two-counter machine}~$M$ consists of a finite set of control
states $Q$, an initial state $q_I \in Q$, a final state $q_F \in Q$, a
set $C$ of counters ($\abs{C} = 2$) and a finite set $\delta_M$ of
instructions manipulating two integer-valued counters. Instructions
are of the form:
\begin{itemize}
\item[$q:$] $c := c + 1$ {\bf goto } $q'$, or
\item[$q:$] {\bf if } $c = 0$  {\bf then goto } $q'$  {\bf else }$c := c - 1$ {\bf goto $q''$}.
\end{itemize}
Formally, instructions are tuples $(q, \alpha, c, q')$ where $q,q' \in
Q$ are source and target states respectively, the action $\alpha \in
\{inc, dec, 0?\}$ applies to the counter $c \in C$.  

A \emph{configuration} of~$M$ is a pair $(q, v)$ where $q \in Q$ and
$v: C \to \nat$ is a valuation of the counters.  An \emph{accepting
  run} of~$M$ is a finite sequence $\pi = (q_0,v_0) \delta_0 (q_1,v_1)
\delta_1 \dots$ $\delta_{n-1} (q_n,v_n)$ where $\delta_i = (q_i,
\alpha_i, c_i, q_{i+1}) \in \delta_M$ are instructions and $(q_i,v_i)$
are configurations of~$M$ such that $q_0 = q_I$, $v_0(c) = 0$ for all
$c \in C$, $q_n = q_F$, and for all $0 \leq i < n$, we have
$v_{i+1}(c) = v_i(c)$ for $c \neq c_i$, and (i)~if $\alpha = inc$,
then $v_{i+1}(c_i) = v_i(c_i) + 1$, (ii)~if $\alpha = dec$, then
$v_i(c_i) \neq 0$ and $v_{i+1}(c_i) = v_i(c_i) - 1$, and (iii)~if
$\alpha = 0?$, then $v_{i+1}(c_i) = v_i(c_i) = 0$.
The \emph{halting problem} asks, given a two-counter machine~$M$,
whether~$M$ has an accepting run.  This problem is
undecidable~\cite{Mi67}.

\paragraph{{\bf Undecidability for RHA with negative rates.}}


Given a two-counter machine~$M$, we construct an RHA $\H_M$ (thus
without diagonal constraints) such that~$M$ has an accepting run if
and only if the answer to the time-bounded reachability problem for
$(\H_M, \goal)$ with time bound $1$ is {\sc Yes}. The construction
of $\H_M$ crucially makes use of both positive and negative rates.

\begin{theorem}
  The time-bounded reachability problem is undecidable for rectangular
  hybrid automata even if restricted to singular rates.
\end{theorem}

\begin{proof}
The reduction is as follows.  The execution steps of $M$ are simulated
in $\H_M$ by a (possibly infinite) sequence of \emph{ticks} within one
time unit. The ticks occur at time $t_0 = 0, t_1 = 1-\frac{1}{4}, t_2
= 1-\frac{1}{16}, \dots$ The counters are encoded as follows.  If the
value of counter $c \in C$ after $i$ execution steps of $M$ is $v(c)$,
then the variable $x_c$ in $\H_M$ has value $\frac{1}{4^{i+v(c)}}$ at
time $t_i$.  Note that this encoding is time-dependent and that the
value of $x_c$ at time $t_i$ is always smaller than $1-t_i =
\frac{1}{4^i}$, and equal to $\frac{1}{4^i}$ if the counter value is
$0$.  To maintain this encoding (if a counter~$c$ is not modified in
an execution step), we need to divide $x_c$ by~$4$ before the next
tick occurs.  We use the divisor gadget in
\figurename~\ref{fig:gadget-div} to do this.  Using the diagram in the
figure, it is easy to check that the value of variable~$x_c$ is
divided by $k^2$ where $k$ is a constant used to define the variable
rates.  In the sequel, we use $k=2$ and $k=4$ (i.e., division by $4$
and by $16$ respectively).  Note also that the division of $\val(x_c)$
by $k^2$ takes $\val(x_c) \cdot (\frac{1}{k} + \frac{1}{k^2})$ time
units, which is less than $\frac{3\cdot \val(x_c)}{4}$ for $k \geq 2$.
Since $\val(x_c) \leq \frac{1}{4^i}$ at step $t_i$, the duration of
the division is at most $\frac{3}{4^i} = t_{i+1} - t_i$, the duration
of the next tick.

We also use the divisor gadget on a variable $x_t$ to construct an
automaton $\atick$ that generates the ticks, as in
\figurename~\ref{fig:gadget-tick}.  We take $k=2$ and we connect and
merge the incoming and outgoing transition of the divisor gadget.
Initially, we require $x_t = 1$. Since division of $x_t$ by $k^2 = 4$
takes $\val(x_t) \cdot (\frac{1}{k} + \frac{1}{k^2}) = \frac{3\cdot
  \val(x_t)}{4}$ time units, it turns out that the value of $x_t$ is
always $1-t_i = \frac{1}{4^i}$ at time $t_i$.  Therefore, we can
produce infinitely many ticks within one time unit.

The automaton $\H_M$ is the product of $\atick$ with the automaton
constructed as follows.  Assume the set of counters is $C= \{c,d\}$.
For each state $q$ of $M$, we construct a location $\l_q$ with rate
$\dot{x}_c = 0$ and $\dot{x}_d = 0$.
For each instruction $(q,\cdot,\cdot,q')$ of $M$, we construct a
transition from location $\l_q$ to $\l_{q'}$ through a synchronized
product of division gadgets to maintain the encoding, as shown in
\figurename~\ref{fig:gadget-inc} and
\figurename~\ref{fig:gadget-zero}. For example, the instruction
$(q,inc,c,q')$ is simulated by dividing $x_c$ by $16 = 4^2$ and $x_d$
by $4$, which transforms for instance $x_c = \frac{1}{4^{i+n}}$ into
$x'_c = \frac{1}{4^{i+n+2}}$.  The decrement is implemented similarly.
Note that the decrement of $c$ requires division by $1$ which is
trivially realized by a location with rate $\dot{x}_c = 0$. Finally,
the zero test is implemented as follows. A counter $c$ has value $0$
in step $i$ if $x_c = 1-t_i = \frac{1}{4^i}$.  Therefore, it suffices
to check that $x_c = x_t$ to simulate a zero test.  To avoid diagonal
constraints, we replace $x_c = x_t$ by a test $x_t = 0$ on the
transition guarded by $x_c = 0$ in the divisor gadget for $x_c$ (as
suggested in \figurename~\ref{fig:gadget-zero}).

The set $\goal = \{\l_{q_F}\}$ contains the location corresponding to
the final state $q_F$ in $M$.  By the above arguments, there is a
one-to-one mapping between the execution of~$M$ and the run of
$\H_M$. In particular, the counter values at step $i$ are correctly
encoded at time $t_i$. Therefore, the location $l_{q_F}$ is reachable
in $\H_M$ within one time unit if and only if $M$ has an accepting run
reaching $q_F$.\qed

\begin{figure}[!tb]
  \begin{center}
    \hrule
    \begin{picture}(130,30)(0,1)


\node[Nframe=n](phantom0)(-4,15){}
\node[Nframe=n](phantom1)(86,15){}

\node[Nmarks=i, Nw=15, Nh=12, Nmr=2](n0)(24,15){$\begin{array}{l@{\,=\,}l}\dot{x} & -k \\ \dot{y} & 1  \end{array}$}
\node[Nmarks=n, Nw=15, Nh=12, Nmr=2](n1)(58,15){$\begin{array}{l@{\,=\,}l}\dot{x} & 1 \\ \dot{y} & -k  \end{array}$}

\node[Nframe=y, dash=1{0}, Nw=72, Nh=20, Nmr=2](box)(41,15){}
\node[Nframe=n](label)(72,7.5){$x/k^2$}



\drawedge[ELpos=50, ELside=l, curvedepth=0](phantom0,n0){$y=0$}

\drawedge[ELpos=50, ELside=l, curvedepth=0](n0,n1){$x=0$}

\drawedge[ELpos=47, ELdist=1, ELside=l, curvedepth=0](n1,phantom1){$y=0$}

\gasset{Nw=0,Nh=0,Nmr=0,loopdiam=6}

\drawline[AHnb=1, dash=1{0}](91,3)(91,30)
\drawline[AHnb=1, dash=1{0}](88,3)(121,3)

\drawline[AHnb=0](90,27)(92,27)

\drawline[AHnb=0](91,27)(103,3)(109,9)
\drawline[AHnb=0](91,3)(103,15)(109,3)

\drawline[AHnb=0, dash={0.2 0.5}0](103,3)(103,15)
\drawline[AHnb=0, dash={0.2 0.5}0](109,3)(109,9)

\node[Nframe=n](label)(96,23){$x$}
\node[Nframe=n](label)(94,9){$y$}

\node[Nframe=n](label)(87,27){$\val(x)$}
\node[Nframe=n](label)(113,11){$\val(x)/k^2$}

\node[Nframe=n](label)(116,5){{\sf time}}


\end{picture}
    \hrule
    \caption{Gadget for division of a variable $x$ by $k^2$. 
The variable $y$ is internal to the gadget. 
The duration of the division is $v \cdot (\frac{1}{k} + \frac{1}{k^2})$. 
The guard ($x_t = 0$) has no influence here, and it is used only when $k=2$.
\label{fig:gadget-div}}

  \end{center}
\end{figure}
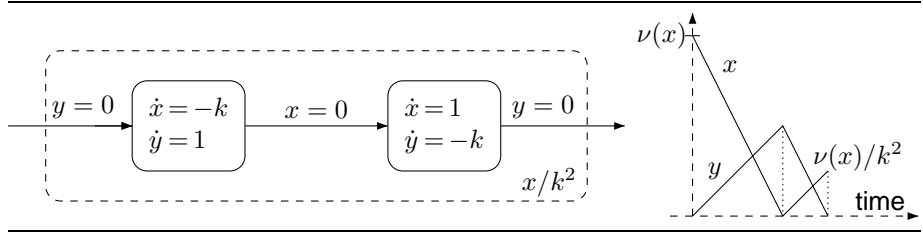

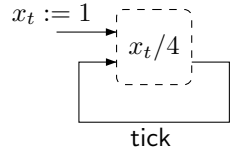
\begin{figure}[!tb]
  \begin{center}
    \hrule
    \begin{picture}(30,22)(0,1)


\node[Nframe=y, dash=1{0}, Nw=10, Nh=10, Nmr=1](box)(15,16){$x_t/4$}

\node[Nframe=n, dash=1{0}, Nw=10, Nh=10, Nmr=1](boxx)(15,18){}
\node[Nframe=n](dummy)(-2,18){}
\drawedge[ELpos=20](dummy,boxx){$x_t:=1$}

\drawline[AHnb=1](20,14)(25,14)(25,6)(5,6)(5,14)(10,14)
\node[Nframe=n](label)(15,4){\tick}

\end{picture}
    \hrule
    \caption{Tick-gadget to produce infinitely many ticks within one time unit.   \label{fig:gadget-tick}}

  \end{center}
\end{figure}

\begin{figure}[!tb]
  \begin{center}
    \hrule
    \begin{picture}(105,22)(0,1)


\node[Nmarks=n, Nw=17, Nh=15, Nmr=2](n0)(10,12){$\begin{array}{cc} q  \\[+2pt] \dot{x_c} = 0 \\ \dot{x_d} = 0  \end{array}$}
\node[Nmarks=n, Nw=17, Nh=15, Nmr=2](n1)(80,12){$\begin{array}{cc} q' \\[+2pt] \dot{x_c} = 0 \\ \dot{x_d} = 0  \end{array}$}
\node[Nframe=n](phantom)(100,12){}
\drawedge[ELpos=65, ELdist=1.5, ELside=l, curvedepth=0](n1,phantom){\tick}

\node[Nmarks=n, Nw=30, Nh=15, Nmr=2](nd)(45,12){}
\node[Nframe=y, dash=1{0}, Nw=10, Nh=10, Nmr=1](box)(37,12){$x_c/16$}
\node[Nframe=n](label)(45,12){$\times$}
\node[Nframe=y, dash=1{0}, Nw=10, Nh=10, Nmr=1](box)(53,12){$x_d/4$}

\drawedge[ELpos=40, ELside=l, curvedepth=0](n0,nd){\tick}
\drawedge[ELpos=50, ELside=l, curvedepth=0](nd,n1){}


\end{picture}
    \hrule
    \caption{Increment-gadget to simulate instruction $(q,inc,c,q')$.  \label{fig:gadget-inc}}

  \end{center}
\end{figure}
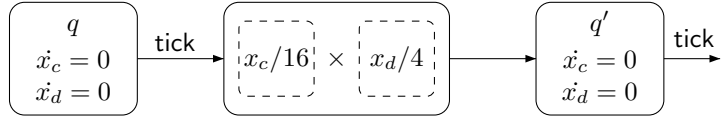

\begin{figure}[!tb]
  \begin{center}
    \hrule
    \begin{picture}(105,22)(0,1)


\node[Nmarks=n, Nw=17, Nh=15, Nmr=2](n0)(10,12){$\begin{array}{cc} q  \\[+2pt] \dot{x_c} = 0 \\ \dot{x_d} = 0  \end{array}$}
\node[Nmarks=n, Nw=17, Nh=15, Nmr=2](n1)(85,12){$\begin{array}{cc} q' \\[+2pt] \dot{x_c} = 0 \\ \dot{x_d} = 0  \end{array}$}
\node[Nframe=n](phantom)(105,12){}
\drawedge[ELpos=65, ELdist=1.5, ELside=l, curvedepth=0](n1,phantom){\tick}

\node[Nmarks=n, Nw=30, Nh=15, Nmr=2](nd)(50,12){}
\node[Nframe=y, dash=1{0}, Nw=10, Nh=10, Nmr=1, NLdist=1](box)(42,12){$x_c/4$}
\node[Nframe=n, dash=1{0}, Nw=10, Nh=10, Nmr=1, NLdist=-3](box)(42,12){{\scriptsize $x_t = 0$}}
\node[Nframe=n](label)(50,12){$\times$}
\node[Nframe=y, dash=1{0}, Nw=10, Nh=10, Nmr=1](box)(58,12){$x_d/4$}

\drawedge[ELpos=40, ELside=l, curvedepth=0](n0,nd){\tick}
\drawedge[ELpos=40, ELside=r, curvedepth=0](n0,nd){($x_c = x_t$)}
\drawedge[ELpos=50, ELside=l, curvedepth=0](nd,n1){}


\end{picture}
    \hrule
    \caption{Zero-gadget to simulate instruction $(q,?0,c,q')$. 
We do use the guard $x_t = 0$ in the divisor gadget for $x_c$,
in order to simulate the diagonal guard $(x_c = x_t)$.  \label{fig:gadget-zero}}

  \end{center}
\end{figure}
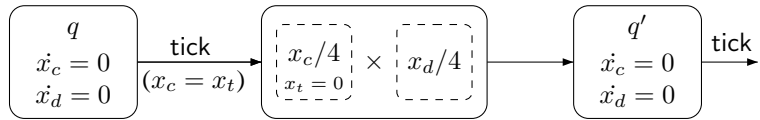

\end{proof}


\paragraph{{\bf Undecidability with diagonal constraints.}}

We now show that diagonal constraints also leads to undecidability.
The result holds even if every variable has a positive, singular,
fixed rate.

\begin{theorem}
  The time-bounded reachability problem is undecidable for LHA that
  use only singular, strictly positive, and fixed-rate variables.
\end{theorem}

\begin{proof}
The proof is again by reduction from the halting problem for
two-counter machines. We describe the encoding of the counters and the
simulation of the instructions.

Given a counter $c$, we represent $c$ via two auxiliary counters
$c_{\mathrm{bot}}$ and $c_{\mathrm{top}}$ such that $v(c) =
v(c_{\mathrm{top}}) - v(c_{\mathrm{bot}})$. 

Incrementing and decrementing $c$ are achieved by 
incrementing either $c_{\mathrm{top}}$ or $c_{\mathrm{bot}}$. 
Zero-testing for $c$ corresponds to checking whether the two auxiliary
counters have the same value. Therefore, we do not need to simulate 
decrementation of a counter.

We encode the value of counter $c_{\mathrm{bot}}$ using two
real-valued variables $x$ and $y$, by postulating that $|x-y| =
\frac{1}{2^{v(c_{\mathrm{bot}})}}$. Both $x$ and $y$ have rate
$\dot{x} = \dot{y} = 1$ at all times and in all locations of the
hybrid automaton.
Incrementing $c_{\mathrm{bot}}$ now simply corresponds to halving the
value of $|x-y|$. In order to achieve this, we use two real-valued 
variables $z$ and $w$ with rate $\dot{z}=2$ and $\dot{w}=3$.

All operations are simulated in `rounds'. At the beginning of a round,
we require that the variables $x,y,z,w$ have respective value
$\frac{1}{2^{v(c_{\mathrm{bot}})}},0,0,0$. We first
explain how we merely \emph{maintain} the value of $c_{\mathrm{bot}}$
throughout a round:
\begin{enumerate}
\item Starting from the beginning of the round, let all variables
  evolve until $x=z$, which we detect via a diagonal
  constraint. Recall that $z$ evolves at twice the rate of $x$.
\item At that point, $x = \frac{2}{2^{v(c_{\mathrm{bot}})}}$ and 
$y = \frac{1}{2^{v(c_{\mathrm{bot}})}}$. Reset $x$ and $z$ to zero.

\item Now let all variables evolve until $y=z$, and reset $y$, $z$ and
  $w$ to zero. It is easy to see that all variables now have exactly the
  same values as they had at the beginning of the round. Moreover, the
  invariant $|x-y| = \frac{1}{2^{v(c_{\mathrm{bot}})}}$ is maintained throughout.
\end{enumerate}
Note that the total duration of the above round is 
$\frac{2}{2^{v(c_{\mathrm{bot}})}}$.
To \emph{increment} $c_{\mathrm{bot}}$, we proceed as
follows:

\begin{enumerate}
\item[$1'$.] Starting from the beginning of the round, let all variables
  evolve until $x=w$. Recall that the rate of $w$ is three times that
  of $x$.
\item[$2'$.] At that point, $x = \frac{1.5}{2^{v(c_{\mathrm{bot}})}}$ and 
$y = \frac{0.5}{2^{v(c_{\mathrm{bot}})}} =
  \frac{1}{2^{v(c_{\mathrm{bot}})+1}}$. 
Reset $x$, $z$, and $w$ to zero.

\item[$3'$.] Now let all variables evolve until $y=z$, and reset $y$, $z$ and
  $w$ to zero. We now have $x = \frac{1}{2^{v(c_{\mathrm{bot}})+1}}$,
  and thus the value of $|x-y|$ has indeed been halved as required.
\end{enumerate}
Note that the total duration of this incrementation round is 
$\frac{1}{2^{v(c_{\mathrm{bot}})}}$, where $v(c_{\mathrm{bot}})$ denotes 
the value of counter $c_{\mathrm{bot}}$ prior to incrementation.

Clearly, the same operations can be simulated for counter
$c_{\mathrm{top}}$ (using further auxiliary real-valued variables).
Note that the durations of the rounds for $c_{\mathrm{bot}}$ and
$c_{\mathrm{top}}$ are in general different---in fact
$c_{\mathrm{bot}}$-rounds are never faster than
$c_{\mathrm{top}}$-rounds.  But because they are powers of
$\frac{1}{2}$, it is always possible to synchronize them, simply by
repeating maintain-rounds for $c_{\mathrm{bot}}$ until the round for
$c_{\mathrm{top}}$ has completed.

Finally, zero-testing the original counter $c$ (which corresponds to
checking whether $c_{\mathrm{bot}} = c_{\mathrm{top}}$) is achieved by
checking whether the corresponding variables have the same value at
the very beginning of a $c_{\mathrm{bot}}$-round (since the
$c_{\mathrm{bot}}$- and $c_{\mathrm{top}}$-rounds are then
synchronized).

We simulate the second counter~$d$ of the machine using further auxiliary counters
$d_{\mathrm{bot}}$ and $d_{\mathrm{top}}$. It is clear that the time
required to simulate one instruction of a two-counter machine is
exactly the duration of the slowest round. Note however that since
counters $c_{\mathrm{bot}}$, $c_{\mathrm{top}}$, $d_{\mathrm{bot}}$,
and $d_{\mathrm{top}}$ are never decremented, the duration of the
slowest round is 
at most $\frac{2}{2^p}$, where $p$ is the
smallest of the initial values of $c_{\mathrm{bot}}$ and
$d_{\mathrm{bot}}$. If a two-counter machine has an accepting run of
length $m$, then the total duration of the simulation is at most~$\frac{2m}{2^p}$.

In order to bound this value, it is necessary before commencing the
simulation to initialize the counters $c_{\mathrm{bot}}$,
$c_{\mathrm{top}}$, $d_{\mathrm{bot}}$, and $d_{\mathrm{top}}$ to a
sufficiently large value, for example any number greater than
$\log_2(m) + 1$. In this way, the duration of the simulation is at
most 1.

Initializing the counters in this way is straightforward.  Starting
with zero counters (all relevant variables are zero) we repeatedly
increment $c_{\mathrm{bot}}$, $c_{\mathrm{top}}$, $d_{\mathrm{bot}}$,
and $d_{\mathrm{top}}$ a nondeterministic number of times, via a
self-loop.  When each of these counters has value $k$, we can
increment all four counters in a single round of duration
$\frac{1}{2^k}$ as explained above.  So over a time period of duration
at most $\sum_{k=0}^\infty \frac{1}{2^k}=2$ the counters can be
initialized to $\lceil \log_2(m) + 1 \rceil$.

Let us now combine these ingredients. Given a two-counter
machine $M$, we construct a hybrid automaton $\H_M$ such that $M$ has
an accepting run iff $\H_M$ has a run of duration at most~3 that reaches 
the final state $\goal$.

$\H_M$ uses the real-valued variables described above to encode the
counters of $M$.  In the initialization phase, $\H_M$
nondeterministically assigns values to the auxiliary counters, hence
guessing the length of an accepting run of~$M$, and then proceeds with
the simulation of $M$.  This ensures a correspondence between an
accepting run of~$M$ and a time-bounded run of $\H_M$ that reaches
$\goal$.\qed
\end{proof}

%



\clearpage

\appendix

\section{Constructions to Prove Proposition~\ref{prop:syntax}\label{sec:constr-prove-prop}}
In this section, we expose three constructions that allow to prove
Proposition~\ref{prop:syntax}. These three constructions have to be
applied successively, starting from an RHA with non-negative rates:
\begin{enumerate}
\item The first construction allows to remove the \emph{non-deterministic
  resets} while preserving time-bounded reachability.
\item The second construction allows to consider only runs where the
  variables are \emph{bounded by $1$}. Roughly speaking, it amounts to
  encode the integral parts of he variables in the locations and adapting
  the guards and invariants accordingly.
\item The third construction allows to consider \emph{strict runs}
  only.
\end{enumerate}

Throughout the section, we assume all the guards to be \emph{reduced},
i.e.: $(i)$ the same atom does not appear twice in the same guard,
$(ii)$ the only guard containing $\true$ is $\true$ and $(iii)$ the
only guard containing $\false$ is $\false$. Remark that any guard can
always be replaced by an equivalent reduced guard. For any valuation
$\val$, we denote by $\val[S/0]$ the valuation s.t. for any $x$:
$\val[S/0](x)=v(x)$ if $x\not\in S$ and $\val[S/0](x)=0$ otherwise.x

\subsection{First construction: deterministic resets}
Given an RHA $\cH$ we show how to construct an RHA $\cH'$ with only
deterministic resets such that $\cH$ is equivalent to $\cH'$ with
respect to reachability in the sense of
Proposition~\ref{prop:eliminate-nd-resets}.  The idea of the
construction is to replace non-deterministic resets in $\cH$ with
resets to $0$ in $\cH'$ and to compensate by suitably altering the
guards of subsequent transitions in $\cH'$.

Let $X=\{x_1,\ldots,x_n\}$ be a set of variables, $\mathcal{I}$ a set of
real intervals including the singleton $\{0\}$, let $g$ be a guard on
$X$, and let $\rho \in \mathcal{I}^n$ be an $n$-tuple of
intervals. (Intuitively $\rho(j)$ represents the interval in which
variable $x_j$ was last reset with $\rho(j)=\{0\}$ if $x_j$ has not yet been
reset.)  Then we inductively define $\adapt{g}{\rho}$ as follows:
\begin{eqnarray*}
\adapt{g_1 \wedge g_2}{\rho} & = & \adapt{g_1}{\rho} \wedge
\adapt{g_2}{\rho}\\
\adapt{x_j \in I}{\rho} & = & x_j \in (I-\rho(j)) \, .
\end{eqnarray*}
Here, given intervals $I,J \subseteq \mathbb{R}$, $I-J$ denotes the
interval $\{ x \mid \exists y\in I, z \in J: x+z=y\}$.

Let $\cH=(X,\loc,\edges,\rate',\invariants,\init)$ be a RHA.  We
construct a new RHA $\dreset{\cH}=(X,\loc',\edges',\rate,\invariants',\init')$
as follows.  Writing $\mathcal{I}$ for the set of intervals used in
variable resets in $\cH$, we have:

\begin{enumerate}
\item $\loc' = \loc \times \mathcal{I}^{|X|}$.
\item For each $\big(\ell,g,r,\ell'\big)\in \edges$ we have that
  $\big((\ell,\rho),g',r',(\ell',\rho')\big)\in \edges'$, where
  $g'=\adapt{g}{\rho}$; $r'(j) = \bot$ and $\rho'(j) = \rho(j)$ if
$r(j) = \bot$; $r'(j)=\{0\}$ and $\rho'(j)=r(j)$ if $r(j) \neq \bot$.
\item $\rate'(\ell,\rho)=\rate(\ell)$.
\item $\invariants'(\ell,\rho)=\adapt{\invariants'(\ell)}{\rho}$.
\item $\init' = \{ (\ell,\mathbf{0}) \mid \ell \in \init\}$, where
  $\mathbf{0}=(\{0\},\ldots,\{0\})$.
\end{enumerate}

\begin{proposition}\label{prop:eliminate-nd-resets}
Let $\ell$ be a location of $\cH$.  Then, $\cH$ admits a
$\tb$-time-bounded run reaching $\ell$ \textbf{iff} $\dreset{\cH}$
admits a $\tb$-time-bounded run reaching some location of the
form~$(\ell,\rho)$.
\end{proposition}

\subsection{Second construction: variables bounded by 1}
Next, we show, given an RHA $\cH$ \emph{with non-negative rates and
  deterministic resets}, how we can build an RHA $\cbound{\cH}$ with
the same properties, and s.t. we can decide time-bounded reachability
on $\cH$ by considering only the runs of $\cbound{\cH}$ with the
variables bounded by $1$.

The idea of the construction is to encode the integer part of the
variable values of $\cH$ in the locations of $\cbound{\cH}$, and to
keep the fractional part (thus, a value in $[0,1]$) in the
variable. To achieve this, locations of $\cbound{\cH}$ are of the form
$(\ell, \mathbf{i})$, where $\ell$ is a location of $\cH$, and
$\mathbf{i}$ is a function that associates a value from
$\{0,\ldots,\cmax\}$ to each variable. Intuitively, $\mathbf{i}(j)$
represents the integer part of $x_j$ in the original run of $\cH'$,
whereas the fractional part is tracked by $x_j$ (hence all the
variables stay in the interval $[0,1]$). For instance, the
configuration $(\ell,2.1,3.2)$ of $\cH$ is encoded by the
configuration $((\ell,(2,3)),0.1,0.2)$ of $\cbound{\cH}$. The
transitions of $\cbound{\cH}$ are adapted from the transitions of
$\cH$ by modifying the guards to take into account the integer part
encoded in the locations. This is achieved thanks to the {\sf Adapt}
function described hereunder. Finally, fresh transitions are added to
$\cbound{\cH}$ that allow to reset variables whose value reach $1$,
while properly adapting the information about the integral part.

Let $X=\{x_1,\ldots,x_n\}$ be a set of variables, let $g$ be a guard on
$X$, and let $\mathbf{i}=(i_1,\ldots,i_n) \in \IN^n$ be a tuple of
natural values. Then, we define inductively
$\adapt{g}{\mathbf{i}}$ as follows:

$$
  \adapt{x_j \le k}{\mathbf{i}}=
  \begin{cases}
    \false & \text{if } k < i_j\\
    x_j = 0 & \text{if } k=i_j\\
    \true & \text{if } k>i_j\\
  \end{cases};
$$

$$
  \adapt{x_j < k}{\mathbf{i}}=
  \begin{cases}
    \false & \text{if } k \leq i_j\\
    x_j < 1 & \text{if } k=i_j+1\\
    \true & \text{if } k>i_j+1\\
  \end{cases};
$$

$$
  \adapt{x_j = k}{\mathbf{i}}=
  \begin{cases}
   \false & \text{if } k < i_j\\
    x_j=0 & \text{if } k=i_j\\
    \false & \text{if } k>i_j\\
  \end{cases};
$$

$$
  \adapt{x_j \ge k}{\mathbf{i}}=
  \begin{cases}
    \false & \text{if } k > i_j+1\\
    x_j = 1 & \text{if } k=i_j+1\\
    \true & \text{if } k\leq i_j\\
  \end{cases};
$$

$$
  \adapt{x_j > k}{\mathbf{i}}=
  \begin{cases}
    \true & \text{if } k < i_j\\
    x_j > 0 & \text{if } k=i_j\\
    \false & \text{if } k>i_j\\
  \end{cases}.
$$

$\adapt{g_1 \wedge g_2}{\mathbf{i}}=\adapt{g_1}{\mathbf{i}} \wedge
\adapt{g_2}{\mathbf{i}}$

Given an RHA $\cH=(X,\loc,\edges,\rate,\invariants,\init)$ s.t. for
any $(\ell, g, r, \ell')\in\edges$, for any $x\in X$: $r(x)$ is either
$[0,0]$ or $\bot$ (that is, all the resets are deterministic and to
zero), we build the RHA
$$\cbound{\cH}=(X,\loc',\edges',\rate',\invariants',\init')$$ as
follows (where $\cmax$ is the largest constant appearing in $\cH$):
\begin{enumerate}
\item $\loc'=\loc\times \{0,\ldots,\cmax\}^n$.
\item  For each $\big(\ell,g,r,\ell'\big)\in \edges$ we have that:
\begin{align*}
  &
  \big((\ell,\mathbf{i}),\adapt{g}{\mathbf{i}},r,(\ell',\mathbf{i'})\big)\in
  \edges', \text{ where } i'_j =
    \begin{cases}
      i_j & \text{if }r(x_j)\neq\bot\\
      0 & \text{otherwise}.
    \end{cases}\\
&
\big((\ell,\mathbf{i}),x_k=1,\{x_k\},(\ell,\mathbf{i'})\big)\in \edges', \text{ where }
      i'_j =
    \begin{cases}
      i_j & \text{if } j \ne k \\ 
      \min(i_j+1,\cmax) & \text{if } j = k.
      \end{cases}
\end{align*}

\item for any $(\ell,i)\in\loc'$: $\rate(\ell,i)=\rate(\ell)$.
\item $\invariants'(\ell,i) = (x_1 \le 1) \wedge \cdots \wedge (x_n
  \le 1)$, for each $(\ell,i)\in\loc'$.
\item $\init'=\big\{(\ell,i)\mid \ell\in\init\big\}$. 
\end{enumerate}

\begin{proposition}\label{prop:bound-construction}
  Let $\cH$ be an RHA with non-negative rates, and s.t. for any edge
  $(\ell, g, r, \ell')$ of $\cH$, for any variable $x$ of $\cH$: $r(x)$
  is either $[0,0]$ or $\bot$. Let $\ell$ be a location of
  $\cH$. Then, $\cH$ admits a $\tb$-time-bounded run reaching $\ell$
  \textbf{iff} $\cbound{\cH}$ admits a $1$-variable-bounded and
  $\tb$-time-bounded run reaching some location of the form
  $(\ell,\mathbf{i})$.
\end{proposition}

\subsection{Third construction: strictly elapsing time}
Last, we explain how we can build an RHA that enforces \emph{strictly
  elapsing time}. Given an RHA
$\cH=(X,\loc,\edges,\rate,\invariants,\init)$ s.t. for any $(\ell, g,
r, \ell')\in\edges$, for any $x\in X$: $r(x)$ is either $[0,0]$ or
$\bot$, we build the RHA
$$\strict{\cH}=(X,\loc',\edges',\rate',\invariants',\init')$$ as
follows. Let $\Pi$ be the (finite) set of all non-empty paths of $\cH$
that contains at most one occurrence of each simple loop. Then:
\begin{enumerate}
\item $\loc'=\loc\times\Pi$
\item $\big((\ell,\pi), g,r,(\ell',\pi')\big)\in\edges'$ iff:
  \begin{itemize}
  \item $\pi=(\ell,g_1,r_1,\ell_1)(\ell_1,g_2,r_2,\ell_2)\ldots(\ell_{n-1},g_n,r_n,\ell')$
  \item $g=\bigwedge_{i=0}^n g_i[X_i/0]$, where $X_i=\{x\mid \exists
    0\leq j<i: r_j(x)\neq\bot\} $
  \item $r$ is s.t. for any $x\in X$: $r(x)=0$ if there is $1\leq
    j\leq n$ s.t. $r(j)\neq \bot$, and $r(x)=\bot$ otherwise.
  \end{itemize}
\item $\rate'$ is s.t. $\rate'(\ell, \pi)=\rate(\ell)$ for any
  $(\ell,\pi)\in\loc'$.
\item $\invariants'$ is s.t.:
  $\invariants'(\ell,\pi)=
  \invariants(\ell)\wedge\bigwedge_{i=1}^n\invariants(\ell_i)[X_i/0]$
  where $X_i=\{x\mid \exists
      0\leq j\leq i: r_j(x)\neq\bot\}$
\item $\init'=\{(\ell,\pi)\mid \ell\in\init\}$.
\end{enumerate}

\begin{proposition}\label{prop:strict-construction}
  Let $\cH$ be an RHA with non-negative rates and s.t. for any edge
  $(\ell, g, r, \ell')$ of $\cH$, for any variable $x$ of $\cH$: $r(x)$
  is either $[0,0]$ or $\bot$. Let $\ell$ be a location of
  $\cH$. Then, $\cH$ admits a $1$-variable-bounded and $\tb$-time-bounded
  run reaching $\ell$ \textbf{iff} $\strict{\cH}$ admits a
  \emph{strict}, $1$-variable-bounded and $\tb$-time-bounded run reaching
  some location of the form~$(\ell,\pi)$.
\end{proposition}

\subsection{Proof of Proposition~\ref{prop:syntax}}
By applying successively the three constructions above to any RHA with
non-negative rates $\cH$, one obtain an RHA
$\cH'=\strict{\cbound{\dreset{\cH}}}$ that has the following
properties: 
\begin{enumerate}
\item $\cH'$ contains only \emph{deterministic resets} to zero
\item All the guards and invariants in $\cH'$ are either $\true$ or
  conjunctions of atoms of the form $x=1$ or $y<1$
  only\footnote{Remark that the third construction  removes from
    the guards all the atoms of the form $x>0$ that are introduced by
    the second one.}. Moreover, each time a variable is tested to $1$ by
  an edge, it is reset to zero.
\end{enumerate}

Moreover, when the original $\cH$ contains no strict inequalities in
the guards and invariants, the same holds for the guards and
invariants of $\cH'$, i.e., they will all be either $\true$ or of the
form $x_1=1\wedge x_2=1\wedge\cdots\wedge x_k=1$ for $\{x_1,\ldots,
x_k\}\subseteq X$. Thus, $\cH'$ has the right syntax, and respects
${\sf H_1}$ through ${\sf H_3}$. Given a location $\ell$ of $\cH$, we
let $\goal$ bet the set of all $\cH'$ locations of the form
$(((\ell,\rho),\mathbf{i}),S)$. 
Thanks to
Proposition~\ref{prop:eliminate-nd-resets}, \ref{prop:bound-construction}
and
\ref{prop:strict-construction}, we are ensured that $\cH$ admits a
$\tb$-time-bounded run reaching $\ell$ iff $\cH'$ admits a strict
$1$-variable-bounded and $\tb$-time-bounded run reaching $\goal$. \qed

\end{document}